\documentclass[a4paper, USenglish, cleveref, autoref, thm-restate]{lipics-v2021}
\hideLIPIcs
\nolinenumbers

\usepackage{color}

\usepackage{tikz}
\usetikzlibrary{positioning, fit, arrows.meta, shapes.geometric, shapes.symbols, decorations.pathreplacing, calc}
\usepackage{svg}

\newcommand{\currRound}{$c$}
\newcommand{\targetRound}{$s$}
\newcommand{\compChain}{\textit{comp-chain}}
\newcommand{\lhChain}{\textit{lh-chain}}
\newcommand{\emptyComp}{\textit{empty-comp}}

\newtheorem{assumption}{Assumption}

\title{The Finality Calculator: Analyzing and Quantifying Filecoin's Finality Guarantees}

\titlerunning{The Finality Calculator}

\author{Guy Goren}{Aptos Labs, USA}{guy.goren@aptoslabs.com}{https://orcid.org/0000-0003-2158-161X}{}
\author{Jorge M. Soares}{Finisterra Labs, USA}{jorge@finisterra.ai}{https://orcid.org/0000-0002-8528-3489}{}

\authorrunning{G. Goren and J. M. Soares}

\Copyright{Guy Goren and Jorge M. Soares} 

\ccsdesc[500]{Security and privacy~Distributed systems security}

\keywords{Blockchain, Filecoin, Stochastic Analysis, Blockchain Security, Finality}

\EventEditors{John Q. Open and Joan R. Access}
\EventNoEds{2}
\EventLongTitle{42nd Conference on Very Important Topics (CVIT 2016)}
\EventShortTitle{CVIT 2016}
\EventAcronym{CVIT}
\EventYear{2016}
\EventDate{December 24--27, 2016}
\EventLocation{Little Whinging, United Kingdom}
\EventLogo{}
\SeriesVolume{42}
\ArticleNo{23}

\begin{document}

\maketitle

\begin{abstract}
    In this paper, we analyze the finality of the Filecoin network, focusing on dynamic probabilistic guarantees of tipset permanence in the canonical chain. Our approach differs from static analyses that consider only the worst-case scenario; instead, we dynamically compute the error probability at each round using the live chain history, providing a more accurate and efficient assessment. 
    We provide a practical algorithm that only requires visibility into the blocks produced by honest participants, which can be implemented by clients or off-chain applications without any change to Filecoin's consensus mechanisms.
    We demonstrate that, under typical operating conditions, the sought-after error probability of $2^{-30}$ is achievable in approximately 30 rounds, a 30x improvement over the 900 rounds that the network currently encodes as a fixed threshold. This finding immediately expedites transactions and enhances usability of the Filecoin network, while laying the foundation for further analysis of other DAG-structured blockchains.
\end{abstract}



\maketitle

\section{Introduction} 

Filecoin is the world's largest decentralized storage network. It leverages a \textit{Proof of Space} mechanism to create an open market for storage, while reusing that same storage in its consensus mechanism. This mechanism provides a more efficient alternative to the energy-intensive Proof of Work model whilst still utilizing a real-world resource, disk space, to back the network's security\footnote{Using a real-world resource for blockchain security ties the network's integrity directly to a physical asset, enhancing its robustness and making attacks costly. This approach is deemed more secure than Proof of Stake, which relies on internal virtual assets that make it more susceptible to manipulation~\cite{ateniese2014proofs,modelingResources22}.}.

Filecoin introduces a consensus mechanism called Storage Power Consensus (SPC), which is based in part on the storage capacity provided by miners, known as Storage Providers (SPs). These providers secure the right to participate in the blockchain's consensus and to create blocks by pledging storage capacity. In return, they receive a financial reward when selected to produce a block that is then successfully added to the chain. The likelihood of a miner producing a block is directly proportional to their power, which is determined by the qualified amount of storage space they have committed and proven through specific cryptographic proofs~\cite{filProofs}.

The SPC mechanism consists of two core components. First, a Sybil-resistance mechanism ensures a reliable and verifiable map of the storage capacity each provider commits. Second, a weighted consensus protocol allows participants with varying storage capacities to contribute to the block production process. This protocol employs a heaviest-chain strategy, similar in spirit to the longest-chain protocols used in other blockchains like Bitcoin~\cite{nakamoto2008bitcoin}, but adapted for Filecoin's structure: instead of a single chain of blocks, Filecoin organizes blocks into \textit{tipsets}, collections of blocks that share the same round and parent. This structure, exemplified in \cref{fig:tipset-diagram}, allows multiple valid blocks submitted in the same round to be grouped together, increasing network throughput and efficiency. The consensus process continually favors the heaviest tipset chain, which is typically the one with the most cumulative storage power backing it.

\begin{figure}[htb]
    \centering
    \resizebox{.6\linewidth}{!}{\begin{tikzpicture}
    \tikzset{
        block/.style = {draw, rectangle, minimum size=0.4cm, align=center, fill=gray!30},
        orphan/.style = {draw, rectangle, minimum size=0.4cm, align=center, fill=red!50},
        tipset/.style = {draw, ellipse, inner sep=2pt, outer sep=0pt, fit=#1, very thick},
        line/.style = {draw, thick, ->},
        redline/.style = {draw, thick, dashed, ->, red},
        tipsetalt/.style = {draw, ellipse, inner sep=5pt, outer sep=0pt, fit=#1, very thick, dashed, draw=red},
        line/.style = {draw, thick, ->}        
    }
    \node[block] (g) {G};

    \node[block, right=2.2cm of g] (a3) {1c};
    \node[block, above=0.2cm of a3] (a2) {1b};
    \node[block, above=0.2cm of a2] (a1) {1a};
    \node[block, below=0.2cm of a3] (a4) {1d};
    \node[block, below=0.2cm of a4] (a5) {1e};

    \node[block, right=2.2cm of a1] (b1) {2a};
    \node[block, below=0.2cm of b1] (b2) {2b};
    \node[block, below=0.2cm of b2] (b3) {2c};
    \node[block, below=0.2cm of b3] (b4) {2d};
    \node[orphan, below=0.2cm of b4] (b5) {2e};

    \node[block, right=2.2cm of b1] (c1) {3a};
    \node[block, below=0.2cm of c1] (c2) {3b};
    \node[block, below=0.2cm of c2] (c3) {3c};
    \node[block, below=0.2cm of c3] (c4) {3d};
    \node[orphan, below=0.2cm of c4] (c5) {3e};

    \node[block, right=2.2cm of c1] (d1) {4a};
    \node[block, below=0.2cm of d1] (d2) {4b};
    \node[block, below=0.2cm of d2] (d3) {4c};
    \node[block, below=0.2cm of d3] (d4) {4d};
    \node[block, below=0.2cm of d4] (d5) {4e};
    
    \node[tipset={(a1) (a2) (a3) (a4) (a5)}, label=above:T1] {};
    \node[tipset={(b1) (b2) (b3) (b4)}, label={[label distance=0.2cm]above:T2}] {};
    \node[tipsetalt={(b1) (b2) (b3) (b4) (b5)}, label={[text=red]below:T2'}] {};
    \node[tipset={(c1) (c2) (c3) (c4)}, label={[label distance=0.2cm]above:T3}] {};
    \node[tipsetalt={(c5)}, label={[text=red]below:T3'}] {};
    \node[tipset={(d1) (d2) (d3) (d4) (d5)}, label=above:T4] {};

    \foreach \from in {a1, a2, a3, a4, a5}
        \draw[line] (\from) -- (g);

    \foreach \from in {b1, b2, b3, b4, b5}
        \foreach \to in {a1, a2, a3, a4, a5}
            \draw[line] (\from) -- (\to);

    \foreach \from in {c1, c2, c3, c4}
        \foreach \to in {b1, b2, b3, b4}
            \draw[line] (\from) -- (\to);
    \foreach \to in {b1, b2, b3, b4, b5}
        \draw[redline] (c5) -- (\to);
    
    \foreach \from in {d1, d2, d3, d4, d5}
        \foreach \to in {c1, c2, c3, c4}
            \draw[line] (\from) -- (\to);
\end{tikzpicture}}
    \caption{Example of a possible tipset chain in Filecoin's block DAG. 
Block~$G$ is the genesis block. In round 1, five blocks are produced, all pointing to~$G$. Round 2 sees five more blocks, each pointing to all blocks from the first round (Tipset T1). In round 3, five blocks are again produced; four point to T2, which contains four blocks, while one (block~$3e$) points to T2', which contains all five blocks. This situation can arise if block~$2e$ is kept private by its producer and shared only with the producer of block~$3e$. Despite block~$3e$ pointing to all blocks of round 2, blocks in round 4 cannot combine it with the rest of round 3 blocks due to different parent tipsets. Additionally, although T2' is heavier than T2, the chain \textit{T1 $\leftarrow$ T2' $\leftarrow$ T3'} is lighter (11 blocks) than the chain \textit{T1 $\leftarrow$ T2 $\leftarrow$ T3} (14 blocks), leading correct nodes to extend the latter.
}
    \label{fig:tipset-diagram}
\end{figure}

Filecoin offers only asymptotic guarantees of correctness. The work in~\cite{wang2023} demonstrates that Filecoin's consensus is safe on an infinite time horizon, provided that the fraction of Byzantine power $f$ satisfies $f\cdot e < 1-\exp\{-(1-f) e\}$, where $e$ represents the expected number of blocks produced per round. With $e~=~5$ in the current Filecoin configuration, this criterion implies that Filecoin can tolerate the presence of a 20\% adversary, given infinite time for system stabilization (i.e. given an infinitely deep block.) However, such theoretical conditions are impractical for real-world applications that require a definitive measure of security. To underscore the need, note that Bitcoin's single-chain-of-blocks model has been extensively analyzed through numerous studies addressing this critical security aspect (e.g.~\cite{gervais2016security,eyal2018majority,sompolinsky2016bitcoin,dembo2020everything,gazi2022practical,guo2022bitcoin}). Filecoin, in contrast, relies on informal arguments for finality~\cite{spec} and a heuristic finality threshold of 900 rounds~\cite{specs-actors}. This approach is notably inadequate for a system whose value reaches into the billions of dollars.

In this work, we introduce an algorithmic approach to quantitatively measure the finality of Filecoin, that is, the probabilistic guarantee that a specific tipset will remain part of the canonical chain and the transactions it contains will not be reversed. Our analysis, which involves distributed computing and stochastic reasoning, establishes dynamic upper bounds for error probabilities --- namely, the likelihood of a chain reorganization overriding a given tipset. Unlike static analyses that only consider worst-case scenarios, our evaluation considers the particular chain structure observed in the rounds surrounding each tipset. This approach allows for more realistic assessments by combining potential future worst-case scenarios with actual past and present conditions.

\paragraph*{Contributions}
The paper delivers several key contributions:
\begin{itemize}
\item It presents the first rigorous quantitative analysis of Filecoin's finality.
\item It introduces a local algorithm, the \textit{Finality Calculator}, which calculates an upper bound on the probability of reorganization based on the recently observed history. Unlike the static confirmation rules frequently used in blockchains, the calculator provides a dynamic measure of error probability that adjusts to real-time chain conditions, decreasing when the relevant history packs more blocks per round and increasing otherwise. This tool operates independently in or alongside a Filecoin node without requiring consensus modifications and does not necessitate additional communication among Filecoin peers. The Finality Calculator has been implemented natively in the most popular Filecoin clients.
\item It provides extensive simulation and real-world evaluation results that demonstrate the effectiveness of the \textit{Finality Calculator}. Our results indicate significant improvements in client-side finality, reducing confirmation times by more than 30x. In most scenarios, the desired $2^{-30}$ upper bound on the probability of reorganization is achieved in fewer than 30 rounds, greatly improving on the previously assumed need for a 900-round wait.
\end{itemize}

\paragraph*{Organization}
The rest of the paper is organized as follows. \Cref{sec:model} outlines the fundamentals of the Filecoin network and the assumptions made in our analysis. \Cref{sec:analysis} then details the analysis from a full node point-of-view, providing a full derivation of the base calculator. \Cref{sec:smart-contract} extends the analysis to the challenging case of an on-chain smart contract with limited visibility. \Cref{sec:implementation} briefly describes our Python prototype implementation, which is used in \cref{sec:simulations,sec:evaluation}. Finally, \cref{sec:conclusions} summarizes this work and lists some directions for future research.

\section{Model}
\label{sec:model}
In Filecoin, participants known as Storage Providers (SPs) pledge storage in exchange for the right to produce blocks. This pledge is quantified in discrete units called sectors, which are then weighed by a Sector Quality Multiplier and provide a commensurate number of voting shares to the controlling SP.
Guaranteeing the correctness of this storage pledging part entails complex cryptographic and incentive mechanisms, which are not in the scope of this work. Instead, we only focus on the parts relevant for our analysis.
We consider a static system comprised of~$N$ sectors with uniform multiplier. A static adversary can corrupt up to a fixed fraction~$f$ of the voting power at the beginning of the protocol, effectively controlling at most~$f\cdot N$ of the total pledged storage. To counter an adaptive adversary capable of dynamically altering their targets -- uncorrupting some nodes while corrupting others during protocol execution (i.e. at the start of each round) -- methods such as key evolving signature schemes~\cite{david2018ouroboros} or checkpointing~\cite{azouvi2022pikachu} could be employed, but this is left for future work.

\paragraph*{Timing}
Filecoin's Storage Power Consensus (SPC) fundamentally depends on synchronous storage proofs and operates within a 30-second interval synchronized by drand~\cite{drand}. For our analysis, we adopt the conventional round-based synchronous model. The Filecoin documentation refers to these intervals as \textit{epochs}~\cite{spec} instead of the more conventional \textit{rounds}, but we choose to use the latter in this work.

\paragraph*{Communication}
The Filecoin network utilizes libp2p~\cite{libp2p} and GossipSub~\cite{vyzovitis2020gossipsub} for peer-to-peer (P2P) and broadcast communication. Our model assumes that the P2P communication network is reliable, meaning that messages sent by a correct participant are guaranteed to reach their destination within the same round.
Additionally, using digital signatures, the integrity of messages is preserved, preventing the adversary from fabricating messages originating from honest participants.
The SPC protocol employs a symmetric consensus mechanism, which avoids the bottlenecks typical of leader-based consensus protocols. All correct Filecoin nodes broadcast messages using a consistent broadcast protocol~\cite{cbroadcast} defined as follows:

\begin{definition}[Consistent broadcast]\label{def:cbcast}
    A protocol for consistent broadcast satisfies:
    \begin{itemize}
        \item[] \textbf{Validity}: If a correct sender $p_s$ consistently broadcasts a message $\textit{msg}$, then all correct processes eventually consistently deliver $\textit{msg}$.
        \item[] \textbf{Consistency}: If a correct process consistently delivers a message $\textit{msg}$ and another correct process consistently delivers a message $\textit{msg}'$, then $\textit{msg}=\textit{msg}'$.
        \item[] \textbf{Integrity}: Each correct process consistently delivers $\textit{msg}$ at most once. Moreover, if the sender $p_s$ is correct, then $\textit{msg}$ was previously consistently broadcast by $p_s$.
    \end{itemize}
\end{definition}

\paragraph*{Randomness}
A random beacon is sourced from drand%
\footnote{The source of randomness in non-PoW blockchains is of considerable importance. The interested reader is advised to see~\cite{drand} for more details.}
and assumed to be both unpredictable and unbiasable by participants in the consensus protocol. The beacon is used to generate seeds for each round, and block producers are selected using a Verifiable Random Function (VRF)~\cite{micali1999verifiable} that assigns a random verifiable score to each participant. Participants whose scores fall below a publicly known threshold are deemed eligible to produce a block in that round. This threshold is adjusted to achieve an expected production (tipset size) of $e = 5$ blocks per round. 
Since each participant independently executes the VRF, the distribution of winners per round can be modeled as a binomial distribution with $N$ trials and a success probability $p=\frac{e}{N}$. Thus, if $X_f[r]$ and $X_h[r]$ represent random variables for the number of blocks won by the adversary and the blocks won by honest participants in round~$r$, respectively, their probabilities are given by:
\begin{align}
    &P(X_f=k) = \textit{Bin}(k;N\cdot f, \frac{e}{N})\label{eq:malBin}\\
    &P(X_h=k) = \textit{Bin}(k;N\cdot (1-f), \frac{e}{N}).\label{eq:honestBin}
\end{align}

Given the large number of SPs ($N>3000$) and the fact that the expected number of blocks per round $e \ll N$, the Poisson approximation of the binomial distribution is appropriate and will be utilized throughout the paper.

\paragraph*{Chain structure}
In Filecoin, a \textit{tipset} is a set of one or more blocks at the same blockchain round that all reference the same set of parent blocks (the parent tipset). Each block within a tipset is treated as a valid part of the blockchain without leading to forks. Duplicated or conflicting transactions within a tipset are excluded based on block producer priority.%
\footnote{A block is valid if the VRF score of its producer is below the required threshold. The priority in transaction inclusion follows these VRF scores, such that out of two duplicated/conflicting transactions, the one belonging to the block with the lower VRF score is included while the other one is excluded. Conflicting transactions within the same block invalidate the entire block.}
This allows multiple blocks to be added to the chain simultaneously, reflecting the concurrent mining efforts across different nodes. \Cref{fig:tipset-diagram} depicts an example of a tipset chain.

\paragraph*{Chain selection}
As in other longest-chain protocols, blocks might be orphaned from the canonical chain if they belong to the lesser side of a fork. Filecoin's fork choice rule states that correct nodes choose to extend the heaviest chain, as defined by a combination of the number of blocks and the storage committed. For simplicity, we conflate this to only the number of blocks.

\section{Analysis: Node view}
\label{sec:analysis}
In this section, we analyze a transaction's finality, as observed by a correct Filecoin node. Specifically, we focus on how a node can estimate the worst-case probability that a past tipset might be reverted based on a node's local block history.

Let~\currRound{} denote the current round and~\targetRound{} the target round for calculating tipset finality. The node calculating finality, denoted $n_i$, bases its calculations on its local history. \lhChain{} denotes the local heaviest chain as observed by~$n_i$ at round~\currRound. 
Similarly, we denote by \compChain$[r]$ the heaviest competitor chain with blocks up to (including) round~$r$, as is observed by~$n_i$ at round~\currRound. In particular, \compChain$[r]$ may be the same as \compChain$[r-1]$ if it does not include blocks produced at round~$r$.
For any given \compChain, We define~$G$ to be the observed (good) advantage that the \lhChain{} holds over \compChain. For example, if \compChain{} contains 4 blocks and \lhChain{} contains 17 blocks, then $G=13$.

\begin{definition}[Finality]\label{def:finality}
    Let $tx$ be a transaction included in~\lhChain. The finality measure of $tx$ is $1-\epsilon(tx)$, where the error probability, $\epsilon(tx)$, is the probability that a reorganization occurs such that the block including $tx$ will be removed from \lhChain.
\end{definition}

Our analysis is divided into three distinct time spans defined by the rounds \targetRound{} and \currRound, for which we provide supporting illustrations in \cref{app:diagrams}: 
\begin{description}
    \item[Distant past] 
    The random variable $L$ quantifies the adversarial lead at round~$s$, representing the excess number of blocks produced by the adversary over those observed in the \lhChain{} up to round~$s$. $L$~is non-negative and equals zero when adversarial chains aren't heavier than the \lhChain. When $L\geq G$, a safety violation is possible.
    \item[Recent past] 
    The random variable $B$ describes the number of blocks produced by the adversary between round~\targetRound{} and current round~\currRound. When $L + B \geq G$, a safety violation is possible.
    \item[Future] 
    The random variable $M$ relates to the (unobserved) future beyond current round \currRound. It describes the number of blocks expected to be produced by the adversary minus the number of blocks produced by honest validators when slowed by the adversary. When $L+B+M \geq G$, a safety violation is possible.
\end{description}

We now prove two lemmas which are significant for our analysis. Colloquially, they establish that it suffices to examine malicious only extensions to all \compChain$[r]$ for $r\in[s,c]$.
\begin{lemma} \label{lem:visibility}
Let $b_h$ be a block produced by an honest node at round $r$. Then the chain ending at the tipset to which~$b_h$ is pointing is known to all honest nodes by round $r+1$.
\end{lemma}
\begin{proof}
Let $h$ denote the creator of block $b_h$. Since $h$ is honest, it broadcasts~$b_h$ to all. The synchrony assumption guarantees the reception of~$b_h$ by round $r+1$. This includes the relevant data to verify~$b_h$, such as the chain of tipsets leading to it.
\end{proof}

\begin{lemma}\label{lem:bestComp}
Let \compChain$[r]$ be a best competitor chain as defined above, and let \textit{ext}-\compChain$[r]$ be an extension of it in the interval $[r,c)$. Let $r'\in[r,c)$ be the latest round at which \textit{ext}-\compChain$[r]$ contains a block produced by an honest node. Then, \compChain$[r']$ contains at least as many blocks as \textit{ext}-\compChain$[r]$ up to round~$r'$.
\end{lemma}
\begin{proof}
Let $b_h$ be a block produced by an honest node in round~$r'$, which is contained in \textit{ext}-\compChain$[r]$.
By \cref{lem:visibility}, block $b_h$ is visible to all honest nodes at time~\currRound. Consequently, the chain ending at a tipset containing $b_h$ is visible to all honest nodes, including $n_i$.
In particular, \textit{ext}-\compChain$[r]$ up to round~$r'$ is visible to~$n_i$. Since \compChain$[r']$ is the heaviest chain visible to~$n_i$ up to round~$r'$, it must contain at least as many blocks as \textit{ext}-\compChain$[r]$ up to round~$r'$.
\end{proof}

\begin{corollary}
The best strategy for an adversary is to extend one of the \compChain s with malicious blocks only.
\end{corollary}
\begin{proof}
    By \cref{lem:visibility} any extension containing blocks by honest nodes is visible to $n_i$, and by \cref{lem:bestComp} such an extension cannot contain more blocks than the corresponding \compChain. Therefore, an adversary that uses its full power to add private blocks to a \compChain{} has the best chance of overtaking the \lhChain{} of $n_i$.
\end{proof}

We can now begin to derive the probabilities for each of the defined time spans. We are interested in the (visible to~$n_i$) \compChain{} that results with the worst (highest) error probability.
For ease of exposition, we detail the calculations for the \compChain$[s]$ that contains no visible orphan blocks, which we denote by \emptyComp.
The calculations for scenarios involving different \compChain s adhere to the same principles but entail additional complexity.
Moreover, note that for a \compChain$[r]$ to have better chances than \emptyComp{} at producing an error, \compChain$[r]$ must have a comparable weight to \emptyComp{} extended with $f\cdot e\cdot (r-s)$ blocks. That is, \compChain$[r]$ should be heavier than the potential private chain the adversary might have created using its full power.%
\footnote{To the best of our knowledge, no \compChain{} longer that 5 epochs had ever achieved that in the history of Filecoin. Indicating that \emptyComp{} is the most relevant in practice.}

\subsection{Span 1: Distant past}\label{sec:distantPast}

Recall that~\targetRound{} is the round for which the finality probability is being evaluated, and~\currRound{} is the current round ($c>s$). The random variable $L$ describes the adversarial lead gained from the last final tipset (e.g. the tipset at round $c-900$) until round~\targetRound. With no observed competition, \Cref{lem:bestComp} establishes that $L$ behaves like a biased random walk with a random step size whenever $L>0$ but does not decrease when $L=0$. For each round $i \in [c-900+1,s]$, the step expectation is $f \cdot e - chain[i]$, where $chain[i]$ is the number of blocks at the tipset of the \lhChain{} that was constructed at round $i$ and $f \cdot e$ is the expected number of adversarial blocks at a round (i.i.d).

The fact that $L$ is non-negative changes the analysis somewhat since we cannot use the classic random walk model. Instead, to account for the distribution of $L$, we can look at a reverse process $L'$ that starts at the tipset of interest of round~\targetRound{} and moves backward in time. Since the adversarial lead can be reached in different number of rounds, the process needs to account for all these possibilities. Specifically, for reaching the lead~$k$ in an attack that lasts~$i$ rounds, we have the random variable~$L_i'$ that follows a binomial distribution
\begin{align}\label{eq:L}
    &L_i' \sim \texttt{Bin}\left(\sum_{j=s-i}^{s} f \cdot n, \frac{e}{n}\right) \approx \texttt{Pois}\left(\sum_{j=s-i}^{s} f \cdot e\right).
\end{align}

For the lead of that attack to be~$k$, it needs to overcome the \lhChain{} during these rounds by~$k$ blocks, that is, it needs to have an advantage of~$k$ over the accumulated blocks in~\lhChain, which we denote by~$k_i$:
\begin{align}
    &k_i = k+\sum_{j=s-i}^{s} chain[j].
\end{align}

It follows that
\begin{align}
    &P(L = k) = P(L' = k) = \max{\left\lbrace P(L_1' =k_1),P(L_2' =k_2),... \right\rbrace}.
\end{align}

\subsection{Span 2: Recent past}\label{sec:recentPast}
The random variable $B$ represents the blocks mined by the adversary between round \targetRound{} and \currRound. \Cref{lem:bestComp} implies that $B$~upper bounds the number of blocks that the adversary could add during these rounds to its hidden chain. Note that $B$ is independent of $L$ and follows a simple binomial distribution, as explained before (see \cref{eq:malBin}). For ease of computation, we again approximate the binomial distribution by a Poisson one: 
\begin{equation}\label{eq:B}
    B \sim \texttt{Bin}\left(\sum_{i=s+1}^{i=c} f \cdot n, \frac{e}{n}\right) \approx \texttt{Pois}\left(\sum_{i=s+1}^{i=c} f \cdot e\right).
\end{equation}

\subsection{Span 3: Future}
\label{sec:future}

The future production of honest blocks follows the binomial distribution as described in \cref{eq:honestBin}. However, it needs to consider that, in cases where the adversary succeeds at splitting the honest chain, not all honest blocks per round would be added to the same tipset. To do so, the adversary must provide parent tipsets that are no worse than the currently available \lhChain{} at correct nodes. We calculate a lower bound on the growth rate of the public chain (shared prefix of \lhChain{} at correct nodes) based on the following two assumptions:
\begin{assumption}
    The adversary can optimally use blocks from round~$i$ to split the blocks created for round~$i+1$.
\end{assumption}
Note that, although the system is probabilistic, this assumption gives the adversary the ability to split it in a deterministically optimal manner. In practice, splitting the network power requires delicate coordination and is extremely difficult. Thus, this assumption considerably strengthens the adversary.
\begin{assumption}
   The adversary uses blocks only from round~$i$ when splitting the blocks created at round $i+1$.
\end{assumption}
This assumption slightly restricts the adversary's capabilities. However, given that the relevance of older blocks diminishes rapidly, the impact of this assumption is minute. We conjecture that the lower bound we establish under these assumptions will also hold without them, as they appear to favor the adversary more than they limit it. However, this claim has not been proven.

We can now rigorously establish the aforementioned lower bound. Recall that Filecoin's underlying broadcast layer only satisfies the properties of consistent broadcast (\cref{def:cbcast}) and does not guarantee reliable broadcast. Thus, the adversary can use blocks from round~$i-1$, denoted by $B[i-1]$, to split the honest chain growth at round~$i$ into $2^{B[i-1]}$ fragments. Denoting the number of honest blocks in round~$i$ as $H[i]$, we get that during round~$i$, the honest chain grows by at least
\begin{equation}
    Z[i] = \min{\left\lbrace \frac{H[i] + B[i-1]}{2^{B[i-1]}}, H[i] \right\rbrace}.
\end{equation}

We therefore have that, at step $i$, the random variable $M$ changes according to the sum $B[i]-Z[i]$. To simplify the calculations, we replace the random variable $Z$ by $Z'$:
\begin{align}
    \begin{split}
        Z'[i] &\sim \texttt{Pois}(E[Z[i]]) \\
        &= \texttt{Pois}\left(E \left[\min{\left\lbrace\frac{H[i]+B[i-1]}{2^{B[i-1]}},H[i]\right\rbrace}\right]\right) \\
        &= \texttt{Pois}\left(\Pr(H[i]>0) \cdot E \left[ \frac{H[i]+B[i-1]}{2^{B[i-1]}}\right]\right).
    \end{split}    
\end{align}

We then define the random process $M_i$ recursively: 
\begin{equation}
    \begin{split}
        M_i &\triangleq M_{i-1} + B[i] -Z'[i], \quad M_0=0 \\
        M_i &= \sum_{j=c+1}^{i}B[j]-Z'[j] = \sum_{j=c+1}^{i}B[j]-\sum_{j=c+1}^{i}Z'[j].
    \end{split}
\end{equation}

For each~$n$ such that $i = c+n$, we have that $\sum_{j=c+1}^{i}B[j] \sim \texttt{Pois}(n \cdot e \cdot f)$ and $\sum_{j=c+1}^{n}Z' \sim \texttt{Pois}(n \cdot E[Z])$.  As the difference between two independent Poisson-distributed random variables, each $M_i$ follows a Skellam distribution~\cite{skellam1946frequency}. Thus, we conclude that:
\begin{align}
    &M_i \sim \texttt{Skellam}(n \cdot e \cdot f, n \cdot E[Z])\\
    &Pr(M=k) = \max{\lbrace Pr(M_1=k), Pr(M_2=k), ...\rbrace}.
\end{align}

\subsection{Error probability}

For an observed good addition $G=k$, the safety violation event  happens only if one of the following three (mutually exclusive) events occurs:

\begin{enumerate}
    \item $L \geq k$
    \item $L < k$ but $L+B \geq k$
    \item $L + B < k$ but $L+B+M \geq k$
\end{enumerate}

Knowing that
\begin{equation}
    P(L+B \ge k \cap L < k) =  \sum_{l=0}^{k-1} P(L=l) \cdot P(B+l \ge k),
\end{equation}
and that
\begin{equation}
    P(L+B+M \ge k \cap L+B<k) = 
    \sum_{l=0}^{k-1}\sum_{b=0}^{k-l-1} P(L=l) \cdot P (B=b) \cdot P(M \ge k - l -b),
\end{equation}
we get 
\begin{equation}
    \begin{aligned}
        &P(\textit{error}) = P(L \ge k) + P(L+B \ge k \cap L < k)
        + P(L+B+M \ge k \cap L+B<k)\\
        &= P(L \ge k) + \sum_{l=0}^{k-1} P(L=l) \cdot  \left(P(B+l \ge k) + \sum_{b=0}^{k-l-1} P(B=b) \cdot P(M \ge k - l - b)\right).
    \end{aligned}
\end{equation}

\section{Analysis: On-chain view}
\label{sec:smart-contract}

Our previous analysis relies on the fact that the node has visibility over all chains that end up with an honest block, as expressed in \cref{lem:visibility,lem:bestComp}.
In this section, we explore the derivation of lower bounds for a finality determination made by a smart contract (\textit{actor} in Filecoin's jargon) that runs within the blockchain and has no access to observations outside that chain. 

As mentioned, a smart contract has no knowledge of honest blocks on other forks, and must therefore estimate finality solely based on the chain in which it runs. We incorporate this difference by considering the possible honest blocks outside the chain as helping the adversarial chain. This affects the calculation of $B$ and $L$, which are now bounded from above as set out below.

These estimates are conservative due to the strict adherence to the correctness of the stochastic steps in our analysis. Achieving tighter and more precise lower bounds would require additional proofs to justify derivations that do not generally hold but may be valid in our specific environment. Addressing these complexities and extending our model accordingly is beyond the scope of this paper and a direction for future research.

\subsection{Step 1: Produced blocks}

Take the total number of blocks produced in round $i$ to be
\begin{align}
    &T[i] = \texttt{Bin}\left(N, \frac{e}{N}\right) \approx \texttt{Pois}\left(e\right)\\
    &T[j,m]\triangleq \sum_{i=j}^{m} T[i] \sim \texttt{Bin}((m-j+1) \cdot N, \frac{e}{N}) \approx \texttt{Pois}\left((m-j+1) \cdot e\right),
\end{align}
and consider the role of the \lhChain{} to now be played by only the observable \textit{chain}. Since we cannot see more blocks than those produced, $T[j,m] \ge chain[j,m] = \sum_{i=j}^{m} chain[i]$. Through the rest of the section and in the interest of brevity, we abuse the notation and overload $T$ to refer to $T[j,m]$ and $chain$ to refer to $chain[j,m]$, following the same logic for related variables. It follows that
\begin{equation}
    \begin{split}
        P(T=k \mid T \ge chain) &= \frac{P(T=k \cap T \ge chain)}{P(T \ge chain)} = 
        \left\{
        \begin{array}{ll}
             \frac{P(T = k)}{P(T \ge chain)} & k \ge chain\\
             0 & otherwise
        \end{array}
        \right.
    \end{split}    
\end{equation}

We can now introduce the random variable $Z$, representing the blocks that are not part of the chain:
\begin{equation}
    Z[j,m] \triangleq T[j,m] - chain[j,m].
\end{equation}
The probability of additional blocks forming part of the adversarial chain is bounded by $P(Z = k \mid chain) = P(T = k + chain \mid T \ge chain)$, given the underlying worst-case assumption that every block outside of \textit{chain} helps the adversarial chain.

\subsection{Step 2: Malicious blocks}

We have that $T[i] = X_f[i] + X_h[i]$, where $X_f$ and $X_h$ are the previously defined variables representing the number of malicious and honest blocks in round $i$. Because malicious blocks $X_f$ may appear on both \textit{chain} and the adversarial chain, we are interested in the joint distribution of the variables $(X_f[j,m], Z[j,m] \mid chain[j,m])$, which we can calculate as 
\begin{equation}\label{eq:bz}
    P(X_f,Z \mid chain) = P(X_f \mid Z, chain) \cdot P(Z \mid chain).
\end{equation}
 As $T = Z + chain$ and $z>0$, we can rewrite \eqref{eq:bz} as
\begin{equation}\label{eq:bz1}
\begin{aligned}
    P(X_f=b, Z=z \mid chain) &= P(X_f=b, T=z+\textit{chain} \mid chain)\\
     &= P (X_f=b \mid T=z+chain, chain) \cdot P(T=z+chain \mid chain).
\end{aligned}
\end{equation}

For the purpose of the analysis, we assume that, given the total number of blocks in a round, the number of malicious blocks is independent of the observed blocks in the chain, that is $P(X_f \mid T, chain) = P(X_f \mid T)$. We therefore rewrite \eqref{eq:bz1} as
\begin{equation}
    P(X_f=b, Z=z \mid chain) =
     P (X_f=b \mid T=z+chain) \cdot P(T=z+chain \mid T \ge chain).
\end{equation}

\subsection{Step 3: Error probability}

We define a new random variable $BpZ$, which provides an upper bound on the adversarial chain by conservatively assuming the sum of all outside blocks and all malicious blocks, resulting in potential double counting (that favors the adversary):
\begin{equation} 
    BpZ[j,m] \triangleq X_f[j,m] + Z[j,m]
\end{equation}

\begin{equation}
\begin{aligned}
        P(BpZ&=k \mid chain) = \sum_{X_f+Z=k} P(X_f, Z \mid chain) \\
        &= \sum_{z=0}^k\left(P(T=z+chain \mid T \ge chain) \vphantom{\sum_{b=k-z}^k} \right.  \left. \cdot \sum_{b=k-z}^k P(X_f = b \mid T=z+chain)\right).
\end{aligned}
\end{equation}

The derivation of the upper bound distributions for $L$ and $B$ follows the same formula as defined in \cref{sec:analysis}, but replacing the Poisson distribution with the distribution of $BpZ$. In particular, the distributions used in the distant past derivation for each $i$-rounds before~\targetRound{} in 
\begin{equation}
    L_i' \sim \texttt{Pois}\left(\sum_{j=s-i}^{s} f \cdot e\right)
    \tag{\ref{eq:L}}
\end{equation}
are replaced by
\begin{equation}\label{eq:Li_actor}
    L_i' \sim BpZ \quad \Longleftrightarrow \quad P(L_i'=k) = P(BpZ=k \mid chain),
    \tag{\ref{eq:L}'}
\end{equation}
and the recent past distribution in
\begin{equation}
    B \sim \texttt{Pois}\left(\sum_{i=s+1}^{i=c} f \cdot e\right)
    \tag{\ref{eq:B}}
\end{equation}
is replaced by
\begin{equation}\label{eq:B_actor}
    B \sim BpZ \quad \Longleftrightarrow \quad P(B=k) = P(BpZ=k \mid chain).
    \tag{\ref{eq:B}'}
\end{equation}

The rest of the calculations for~$L$ and~$B$ continue according to the derivations in \cref{sec:distantPast,sec:recentPast} respectively.
In addition, the calculation of $M$ and the final determination of the error probability remain unchanged.

\section{Implementation}
\label{sec:implementation}

We implemented the exact algorithms defined by the formulas in \cref{sec:analysis,sec:smart-contract} as a set of Python functions, leveraging NumPy and SciPy for probability computations. The code is open source, released under Apache 2.0 and MIT, and is available on GitHub~\cite{ecfc-github}. The repository also includes all the simulation and real-world traces used in this paper, as well as the code to process it and generate the figures.

The implementation, which we name the \textit{Finality Calculator}, can compute finality both from the node's perspective and from the smart contract's perspective and serves a dual purpose: in addition to enabling the validation and evaluation of the algorithms, it can be leveraged by any Filecoin user to determine the error probability for a given block based on a chain trace. However, this code has not been optimized for performance and does not currently integrate with client software to automatically fetch said traces. Optimized implementations are available in both Lotus\footnote{https://github.com/filecoin-project/lotus/pull/13547} and Forest\footnote{https://github.com/ChainSafe/forest/pull/6785}, the most widely used Filecoin clients. These production implementations are fully integrated with their respective clients, removing the need to manually supply chain traces and enabling real-time finality estimation for network participants.

The algorithms, as specified, include several operations over an unbounded range of the advantage $k$. In the implementation, these were truncated by default to a maximum $max\_k_{L,B}=400$ and $max\_k_M=100$, paired with an early stop condition when the probability of an event attains $10^{-25}$. Moreover, intermediate searches over neighboring rounds were also limited to $max\_i_{L} = 25$ and $max\_i_{M} = 100$. Further relaxing these parameters yields no substantive increase in accuracy for settlement times up to 100 rounds, whereas restricting them further is unnecessary due to the early stop condition.

\section{Simulations}
\label{sec:simulations}

We conduct a first empirical study of our algorithm in simulation. Recall that the Filecoin network is set to produce, on expectation, 5 blocks per round. We introduce the notion of \textit{chain fullness} $\alpha$ to represent the ratio of average blocks per tipset in the \lhChain{} to this target number. We generate synthetic chain traces using a Poisson distribution whose parameter $\lambda=\alpha \cdot 5$ is the expected number of blocks per tipset. For a perfectly healthy chain with $\alpha=1.00$, the average number of blocks in a tipset is 5 -- the target adopted as a network parameter. For each value of $\alpha$, we generate 7 independent traces. Specifically, we run experiments for $\alpha$ ranging from $0.80$ to $1.00$ (4 to 5 blocks per tipset on average) and generate 40,000 rounds of chain history for each run. The chain traces are then fed to the calculator using different settlement times (in rounds). Except where otherwise noted, we use the default calculator parameters in \cref{sec:implementation} and a Byzantine fraction $f = 0.3$ (a system-level assumption).

\Cref{fig:simulation_trends} shows the results of running the calculator on a simulated chain in typical fullness — $\alpha=0.96$, the 24-hour average value observed on the Filecoin network at the time of writing (as retrieved from Filscan~\cite{filscan}), and a fairly typical value. On (a), we present the results for the node case, where we see that the error probability drops exponentially with the increase in settlement time, before settling around $10^{-25}$ due to the early stop condition in the implementation. Notably, an error probability of $10^{-10}$---corresponding to a once-in-10,000-years event---is attained in fewer than 30 rounds, already substantially smaller than the $2^{-30} \approx 10^{-9}$ targeted by the 900-round soft finality parameter in Filecoin. Moving to plot (b), we see that the on-chain algorithm requires approximately 60 rounds to attain the same certainty that the node attains in less than 30 rounds, illustrating the limitations imposed by the limited information available in the on-chain case.


\begin{figure}
    \centering
    \begin{subfigure}[b]{0.24\textwidth}
        \centering
        \includesvg[inkscapelatex=false,width=\textwidth]{figures/simulation/trend_validator.svg}
        \caption{Node}
        \label{fig:trend_validator}
    \end{subfigure}
    \begin{subfigure}[b]{0.24\textwidth}
        \centering
        \includesvg[inkscapelatex=false,width=\textwidth]{figures/simulation/trend_actor.svg}
        \caption{Smart contract}
        \label{fig:trend_actor}
    \end{subfigure}
    \hfill
    \begin{subfigure}[b]{0.24\textwidth}
        \centering
        \includesvg[inkscapelatex=false,width=\textwidth]{figures/simulation/scatter_validator_20_sample.svg}
        \caption{20 rounds}
        \label{fig:scatter_validator_20}
    \end{subfigure}
    \begin{subfigure}[b]{0.24\textwidth}
        \centering
        \includesvg[inkscapelatex=false,width=\textwidth]{figures/simulation/scatter_validator_40_sample.svg}
        \caption{40 rounds}
        \label{fig:scatter_validator_40}
    \end{subfigure}    
    \caption{(a) and (b): Error probabilities for a range of settlement times, from the perspective of a node and of a smart contract, respectively, using simulation data for $\alpha=0.96$. The plot shows the median results over $7$ runs. Note that the $y$ axes use a logarithmic scale with different ranges. (c) and (d): Error probabilities from the perspective of a node, after a settlement time of 20 and 40 rounds, respectively, and for different chain fullness conditions. Each point represents a single calculator output, and the linear trend line is plotted in red.}
    \label{fig:simulation_trends}        
\end{figure}

We zoom in on the node case in plots (c) and (d), which presents the error probabilities for different levels of chain fullness, with the trend line in red. Non-full chains, as expected, lead to a much higher error probability and make it impossible to attain the same level of certainty without increasing the settlement time. Nevertheless, even in cases of far-from-full chains, the $10^{-10}$ threshold is crossed in under 40 rounds. 


Lastly, \cref{fig:simulation_details} shows the computed error probability over single runs of simulated full and non-full chains. While the smart contract and node error probabilities are correlated, and the node calculator always attains higher certainty, the two algorithms operate on somewhat different data and, therefore, their outputs are not simply scaled versions of the same curve.
\begin{figure}
    \centering
    \begin{subfigure}[b]{0.40\textwidth}
        \centering
        \includesvg[inkscapelatex=false,width=\textwidth]{figures/simulation/80_0_30.svg}
        \caption{$\alpha$ = 0.80}
        \label{fig:error_80}
    \end{subfigure}
    \quad
    \begin{subfigure}[b]{0.40\textwidth}
        \centering
        \includesvg[inkscapelatex=false,width=\textwidth]{figures/simulation/100_0_30.svg}
        \caption{$\alpha$ = 1.00}
        \label{fig:error_100}
    \end{subfigure}
    \caption{Error probabilities from the perspective of a node and of a smart contract, after a settlement time of 30 rounds, for single runs with $\alpha = 0.80$ and $\alpha = 1.00$. The sampled tipset sizes are plotted in green and the line shows the 30-round moving average.}
    \label{fig:simulation_details}        
\end{figure}

\section{Evaluation}\label{sec:evaluation}

We further evaluated our work on real-world Filecoin chain traces obtained from an archival node. The traces were collected in 2023 and contain $80,000$ sequential round numbers and the count of blocks included in the corresponding round. They were fed into the same calculator implementation used in the simulation study.

\Cref{fig:evaluation_details} presents the results of running the calculator with settlement time $30$ on the aforementioned traces. The plot on the left shows the tipset sizes under normal operating conditions, and the resulting relatively low and stable error probabilities. The plot on the right covers a period of normal operation, with typical chain fullness; then, approaching round $2,740,000$, an incident causes the chain to go through several days of lower block production, during which the finality estimates correspondingly adapt. This captures the advantage of the finality calculator in a real-world setting: it enables users to target a desired level of safety and attain it in the shortest possible time under noraml operation, while graciously but rigorously retaining the target guarantees by falling back to longer finality delays in the presence of disturbances.
\begin{figure}
    \centering
    \begin{subfigure}[b]{0.40\textwidth}
        \centering
        \includesvg[inkscapelatex=false,width=\textwidth]{figures/evaluation/november_30.svg}
        \caption{Typical chain (Nov 2023)}
        \label{fig:plot_eval_nov}
    \end{subfigure}
    \quad
    \begin{subfigure}[b]{0.40\textwidth}
        \centering
        \includesvg[inkscapelatex=false,width=\textwidth]{figures/evaluation/march_30.svg}
        \caption{Incident chain (Mar-Apr 2023)}
        \label{fig:plot_eval_mar}
    \end{subfigure}
    \caption{Error probabilities from the perspective of a node and of a smart contract, after a settlement time of 30 rounds, for two real-world chain traces. The sampled tipset sizes are plotted in green and the line shows the 30-round moving average.}
    \label{fig:evaluation_details}        
\end{figure}

\section{Discussion}\label{sec:conclusions}
This research is the first quantitative analysis of Filecoin’s finality metrics, combining theoretical derivations with empirical analysis through simulations and real data. Our findings significantly advance the understanding of finality, showing that the desired error probability of~$2^{-30}$ is typically achieved within just 30 rounds. This marks a substantial improvement over the previously assumed necessity of 900 rounds, offering users faster settlement times and thus enhancing the user experience.

A key conceptual advance in this study is the shift from static, worst-case analyses to a dynamic approach that utilizes real-time chain data to inform current error probability calculations. This method refines the accuracy of predictions and directly impacts practical applications by providing more reliable metrics for network participants.


These insights open up several avenues for future research, notably in developing tighter bounds for on-chain finality, which could further enhance the user experience. This study lays a foundational step towards understanding and improving the mechanisms of finality within Filecoin and the broader class of layered DAG-structured blockchains. The implications of this work are expected to inform both theoretical advancements and immediate practical improvements in blockchain networks. 

A portion of this work formed the basis for a merged Filecoin standards proposal~\cite{frc0089}, which has since been implemented in both Lotus and Forest, the two leading Filecoin clients, enabling faster transaction finality in practice and demonstrating the real-world impact of this work.


\bibliography{references}

\clearpage
\appendix

\section{Illustrations of time spans}
\label{app:diagrams}

\begin{figure}[htb]
    \centering
\scalebox{0.45}{\begin{tikzpicture}[block/.style={draw, rectangle, minimum size=0.2cm, align=center, fill=gray!30},
                    tipset/.style={draw, ellipse, inner sep=2pt, fit=#1, thick},
                    line/.style={draw, thick, ->},
                    orphan/.style={draw, rectangle, minimum size=0.2cm, align=center, fill=red!50},
                    tipsetalt/.style={draw, ellipse, inner sep=2pt, fit=#1, thick, dashed, red},
                    tipsethighlight/.style={draw, ellipse, inner sep=2pt, fit=#1, very thick},
                    chainline/.style={draw, thick, dashed, ->, red},
                    compcloud/.style={draw, cloud, cloud puffs=11, cloud ignores aspect, inner sep=-4pt, fit=#1, very thick, red}]

    
    \node[block] at (0,0) (s1) {};
    \node[block, below=0.1cm of s1] (s2) {};
    \node[block, below=0.1cm of s2] (s3) {};
    \node[block, below=0.1cm of s3] (s4) {};
    \node[tipsethighlight={(s1) (s2) (s3) (s4)}, label=above: $s$] (ts) {};
    
    \node[block, right=0.6cm of s1] (b11) {};
    \node[block, below=0.1cm of b11] (b12) {};
    \node[block, below=0.1cm of b12] (b13) {};
    \node[block, below=0.1cm of b13] (b14) {};
    \node[block, below=0.1cm of b14] (b15) {};
    \node[tipset={(b11) (b12) (b13) (b14) (b15)}] (b1) {};

    \node[block, right=0.6cm of b11] (b21) {};
    \node[block, below=0.1cm of b21] (b22) {};
    \node[block, below=0.1cm of b22] (b23) {};
    \node[block, below=0.1cm of b23] (b24) {};
    \node[tipset={(b21) (b22) (b23) (b24)}] (b2) {};

    \node[block, right=0.6cm of b21] (b31) {};
    \node[block, below=0.1cm of b31] (b32) {};
    \node[block, below=0.1cm of b32] (b33) {};
    \node[block, below=0.1cm of b33] (b34) {};
    \node[block, below=0.1cm of b34] (b35) {};
    \node[block, below=0.1cm of b35] (b36) {};
    \node[tipset={(b31) (b32) (b33) (b34) (b35) (b36)}] (b3) {};

    \node[block, right=0.6cm of b31] (b41) {};
    \node[block, below=0.1cm of b41] (b42) {};
    \node[block, below=0.1cm of b42] (b43) {};
    \node[block, below=0.1cm of b43] (b44) {};
    \node[block, below=0.1cm of b44] (b45) {};
    \node[tipset={(b41) (b42) (b43) (b44) (b45)}] (b4) {};

    \node[block, right=0.6cm of b41] (b51) {};
    \node[block, below=0.1cm of b51] (b52) {};
    \node[block, below=0.1cm of b52] (b53) {};
    \node[block, below=0.1cm of b53] (b54) {};
    \node[tipset={(b51) (b52) (b53) (b54)}] (b5) {};
    
    \node[block, right=0.6cm of b51] (c1) {};
    \node[block, below=0.1cm of c1] (c2) {};
    \node[block, below=0.1cm of c2] (c3) {};
    \node[block, below=0.1cm of c3] (c4) {};
    \node[block, below=0.1cm of c4] (c5) {};
    \node[tipsethighlight={(c1) (c2) (c3) (c4) (c5)}, label=above: $c$] (tc) {};

    \node[block, left=0.6cm of s1] (bp11) {};
    \node[block, below=0.1cm of bp11] (bp12) {};
    \node[block, below=0.1cm of bp12] (bp13) {};
    \node[tipset={(bp11) (bp12) (bp13)}] (bp1) {};

    \node[block, left=0.6cm of bp11] (bp21) {};
    \node[block, below=0.1cm of bp21] (bp22) {};
    \node[block, below=0.1cm of bp22] (bp23) {};
    \node[block, below=0.1cm of bp23] (bp24) {};
    \node[block, below=0.1cm of bp24] (bp25) {};
    \node[tipset={(bp21) (bp22) (bp23) (bp24) (bp25)}] (bp2) {};

    \node[block, left=0.6cm of bp21] (bp31) {};
    \node[block, below=0.1cm of bp31] (bp32) {};
    \node[block, below=0.1cm of bp32] (bp33) {};
    \node[block, below=0.1cm of bp33] (bp34) {};
    \node[block, below=0.1cm of bp34] (bp35) {};
    \node[tipset={(bp31) (bp32) (bp33) (bp34) (bp35)}] (bp3) {};

    \node[block, left=0.6cm of bp31] (bp41) {};
    \node[block, below=0.1cm of bp41] (bp42) {};
    \node[block, below=0.1cm of bp42] (bp43) {};
    \node[block, below=0.1cm of bp43] (bp44) {};
    \node[block, below=0.1cm of bp44] (bp45) {};
    \node[tipset={(bp41) (bp42) (bp43) (bp44) (bp45)}] (bp4) {};

    \node[block, left=0.6cm of bp41] (bp51) {};
    \node[block, below=0.1cm of bp51] (bp52) {};
    \node[block, below=0.1cm of bp52] (bp53) {};
    \node[block, below=0.1cm of bp53] (bp54) {};
    \node[block, below=0.1cm of bp54] (bp55) {};
    \node[tipset={(bp51) (bp52) (bp53) (bp54) (bp55)}] (bp5) {};

    \node [left=0.6cm of bp5] (dots) {\ldots};

    \node[orphan, above=1.9cm of bp31] (f11) {};
    \node[orphan, above=0.1cm of f11] (f12) {};
    \node[orphan, above=0.1cm of f12] (f13) {};
    \node[orphan, above=0.1cm of f13] (f14) {};
    \node[orphan, above=0.1cm of f14] (f15) {};
    \node[orphan, above=0.1cm of f15] (f16) {};
    \node[orphan, above=0.1cm of f16] (f17) {};
    \node[tipsetalt={(f11) (f12) (f13) (f14) (f15) (f16) (f17)}] (f1) {};

    \node[orphan, right=0.6cm of f11] (f21) {};
    \node[orphan, above=0.1cm of f21] (f22) {};
    \node[orphan, above=0.1cm of f22] (f23) {};
    \node[orphan, above=0.1cm of f23] (f24) {};
    \node[tipsetalt={(f21) (f22) (f23) (f24)}] (f2) {};

    \node[orphan, right=0.6cm of f21] (f31) {};
    \node[orphan, above=0.1cm of f31] (f32) {};
    \node[orphan, above=0.1cm of f32] (f33) {};
    \node[orphan, above=0.1cm of f33] (f34) {};
    \node[orphan, above=0.1cm of f34] (f35) {};
    \node[tipsetalt={(f31) (f32) (f33) (f34) (f35)}] (f3) {};
    
    \draw[line] (bp5) -- (dots);
    \draw[line] (bp4) -- (bp5);
    \draw[line] (bp3) -- (bp4);
    \draw[line] (bp2) -- (bp3);
    \draw[line] (bp1) -- (bp2);
    \draw[line] (ts) -- (bp1);
    \draw[line] (b1) -- (ts);
    \draw[line] (b2) -- (b1);
    \draw[line] (b3) -- (b2);
    \draw[line] (b4) -- (b3);
    \draw[line] (b5) -- (b4);
    \draw[line] (tc) -- (b5);
    
    \draw[chainline] (f1) -- (bp4);
    \draw[chainline] (f2) -- (f1);
    \draw[chainline] (f3) -- (f2);
  
    \node[compcloud={(f1) (f2) (f3)}, label={[text=red]above:$L=16-13=3$}] (cloud1) {};
    
    \draw [decorate,decoration={brace,amplitude=10pt,mirror,raise=5pt}, draw=green!40!black]
    ($(ts.south west) - (0,1)$) -- ($(tc.south east) - (0,0.7)$) 
    node [pos=0.5,anchor=north,yshift=-0.55cm, text=green!40!black] {good addition$=33$};
\end{tikzpicture}}
\caption{A diagram illustrating the meaning of~$L$ when accounting for the distant past.
}
    \label{fig:distantPast}
\end{figure}

\begin{figure}[htb]
    \centering
\scalebox{0.45}{\begin{tikzpicture}[block/.style={draw, rectangle, minimum size=0.2cm, align=center, fill=gray!30},
                    tipset/.style={draw, ellipse, inner sep=2pt, fit=#1, thick},
                    line/.style={draw, thick, ->},
                    orphan/.style={draw, rectangle, minimum size=0.2cm, align=center, fill=red!50},
                    tipsetalt/.style={draw, ellipse, inner sep=2pt, fit=#1, thick, dashed, red},
                    tipsethighlight/.style={draw, ellipse, inner sep=2pt, fit=#1, very thick},
                    chainline/.style={draw, thick, dashed, ->, red},
                    compcloud/.style={draw, cloud, cloud puffs=11, cloud ignores aspect, inner sep=-4pt, fit=#1, very thick, red}]

    
    \node[block] at (0,0) (s1) {};
    \node[block, below=0.1cm of s1] (s2) {};
    \node[block, below=0.1cm of s2] (s3) {};
    \node[block, below=0.1cm of s3] (s4) {};
    \node[tipsethighlight={(s1) (s2) (s3) (s4)}, label=above: $s$] (ts) {};
    
    \node[block, right=0.6cm of s1] (b11) {};
    \node[block, below=0.1cm of b11] (b12) {};
    \node[block, below=0.1cm of b12] (b13) {};
    \node[block, below=0.1cm of b13] (b14) {};
    \node[block, below=0.1cm of b14] (b15) {};
    \node[tipset={(b11) (b12) (b13) (b14) (b15)}] (b1) {};

    \node[block, right=0.6cm of b11] (b21) {};
    \node[block, below=0.1cm of b21] (b22) {};
    \node[block, below=0.1cm of b22] (b23) {};
    \node[block, below=0.1cm of b23] (b24) {};
    \node[tipset={(b21) (b22) (b23) (b24)}] (b2) {};

    \node[block, right=0.6cm of b21] (b31) {};
    \node[block, below=0.1cm of b31] (b32) {};
    \node[block, below=0.1cm of b32] (b33) {};
    \node[block, below=0.1cm of b33] (b34) {};
    \node[block, below=0.1cm of b34] (b35) {};
    \node[block, below=0.1cm of b35] (b36) {};
    \node[tipset={(b31) (b32) (b33) (b34) (b35) (b36)}] (b3) {};

    \node[block, right=0.6cm of b31] (b41) {};
    \node[block, below=0.1cm of b41] (b42) {};
    \node[block, below=0.1cm of b42] (b43) {};
    \node[block, below=0.1cm of b43] (b44) {};
    \node[block, below=0.1cm of b44] (b45) {};
    \node[tipset={(b41) (b42) (b43) (b44) (b45)}] (b4) {};

    \node[block, right=0.6cm of b41] (b51) {};
    \node[block, below=0.1cm of b51] (b52) {};
    \node[block, below=0.1cm of b52] (b53) {};
    \node[block, below=0.1cm of b53] (b54) {};
    \node[tipset={(b51) (b52) (b53) (b54)}] (b5) {};
    
    \node[block, right=0.6cm of b51] (c1) {};
    \node[block, below=0.1cm of c1] (c2) {};
    \node[block, below=0.1cm of c2] (c3) {};
    \node[block, below=0.1cm of c3] (c4) {};
    \node[block, below=0.1cm of c4] (c5) {};
    \node[tipsethighlight={(c1) (c2) (c3) (c4) (c5)}, label=above: $c$] (tc) {};

    \node[block, left=0.6cm of s1] (bp11) {};
    \node[block, below=0.1cm of bp11] (bp12) {};
    \node[block, below=0.1cm of bp12] (bp13) {};
    \node[tipset={(bp11) (bp12) (bp13)}] (bp1) {};

    \node[block, left=0.6cm of bp11] (bp21) {};
    \node[block, below=0.1cm of bp21] (bp22) {};
    \node[block, below=0.1cm of bp22] (bp23) {};
    \node[block, below=0.1cm of bp23] (bp24) {};
    \node[block, below=0.1cm of bp24] (bp25) {};
    \node[tipset={(bp21) (bp22) (bp23) (bp24) (bp25)}] (bp2) {};

    \node[block, left=0.6cm of bp21] (bp31) {};
    \node[block, below=0.1cm of bp31] (bp32) {};
    \node[block, below=0.1cm of bp32] (bp33) {};
    \node[block, below=0.1cm of bp33] (bp34) {};
    \node[block, below=0.1cm of bp34] (bp35) {};
    \node[tipset={(bp31) (bp32) (bp33) (bp34) (bp35)}] (bp3) {};

    \node[block, left=0.6cm of bp31] (bp41) {};
    \node[block, below=0.1cm of bp41] (bp42) {};
    \node[block, below=0.1cm of bp42] (bp43) {};
    \node[block, below=0.1cm of bp43] (bp44) {};
    \node[block, below=0.1cm of bp44] (bp45) {};
    \node[tipset={(bp41) (bp42) (bp43) (bp44) (bp45)}] (bp4) {};

    \node[block, left=0.6cm of bp41] (bp51) {};
    \node[block, below=0.1cm of bp51] (bp52) {};
    \node[block, below=0.1cm of bp52] (bp53) {};
    \node[block, below=0.1cm of bp53] (bp54) {};
    \node[block, below=0.1cm of bp54] (bp55) {};
    \node[tipset={(bp51) (bp52) (bp53) (bp54) (bp55)}] (bp5) {};

    \node [left=0.6cm of bp5] (dots) {\ldots};

    \node[orphan, above=1.9cm of ts] (f11) {};
    \node[orphan, above=0.1cm of f11] (f12) {};
    \node[orphan, above=0.1cm of f12] (f13) {};
    \node[orphan, above=0.1cm of f13] (f14) {};
    \node[tipsetalt={(f11) (f12) (f13) (f14)}] (f1) {};

    \node[orphan, right=0.6cm of f11] (f21) {};
    \node[orphan, above=0.1cm of f21] (f22) {};
    \node[orphan, above=0.1cm of f22] (f23) {};
    \node[orphan, above=0.1cm of f23] (f24) {};
    \node[orphan, above=0.1cm of f24] (f25) {};
    \node[tipsetalt={(f21) (f22) (f23) (f24) (f25)}] (f2) {};

    \node[orphan, right=0.6cm of f21] (f31) {};
    \node[orphan, above=0.1cm of f31] (f32) {};
    \node[orphan, above=0.1cm of f32] (f33) {};
    \node[tipsetalt={(f31) (f32) (f33)}] (f3) {};

    \node[orphan, right=0.6cm of f31] (f41) {};
    \node[orphan, above=0.1cm of f41] (f42) {};
    \node[orphan, above=0.1cm of f42] (f43) {};
    \node[orphan, above=0.1cm of f43] (f44) {};
    \node[orphan, above=0.1cm of f44] (f45) {};
    \node[tipsetalt={(f41) (f42) (f43) (f44) (f45)}] (f4) {};
    
    \node[orphan, right=0.6cm of f41] (f51) {};
    \node[orphan, above=0.1cm of f51] (f52) {};
    \node[orphan, above=0.1cm of f52] (f53) {};
    \node[tipsetalt={(f51) (f52) (f53)}] (f5) {};

    \node[orphan, right=0.6cm of f51] (f61) {};
    \node[orphan, above=0.1cm of f61] (f62) {};
    \node[orphan, above=0.1cm of f62] (f63) {};
    \node[tipsetalt={(f61) (f62) (f63)}] (f6) {};
    
    \node[orphan, right=0.6cm of f61] (f71) {};
    \node[orphan, above=0.1cm of f71] (f72) {};
    \node[tipsetalt={(f71) (f72)}] (f7) {};
    
    \draw[line] (bp5) -- (dots);
    \draw[line] (bp4) -- (bp5);
    \draw[line] (bp3) -- (bp4);
    \draw[line] (bp2) -- (bp3);
    \draw[line] (bp1) -- (bp2);
    \draw[line] (ts) -- (bp1);
    \draw[line] (b1) -- (ts);
    \draw[line] (b2) -- (b1);
    \draw[line] (b3) -- (b2);
    \draw[line] (b4) -- (b3);
    \draw[line] (b5) -- (b4);
    \draw[line] (tc) -- (b5);
    
    \draw[chainline] (f1) -- (bp1);
    \draw[chainline] (f2) -- (f1);
    \draw[chainline] (f3) -- (f2);
    \draw[chainline] (f4) -- (f3);
    \draw[chainline] (f5) -- (f4);
    \draw[chainline] (f6) -- (f5);
    \draw[chainline] (f7) -- (f6);

    \node[compcloud={(f1) (f2) (f3) (f4) (f5) (f6) (f7)}, label={[text=red]above:$B=25$}] (cloud1) {};
    
    \draw [decorate,decoration={brace,amplitude=10pt,mirror,raise=5pt}, draw=green!40!black]
    ($(ts.south west) - (0,1)$) -- ($(tc.south east) - (0,0.7)$) 
    node [pos=0.5,anchor=north,yshift=-0.55cm, text=green!40!black] {good addition$=33$};
\end{tikzpicture}}
\caption{A diagram illustrating the meaning of~$B$ when accounting for the recent past.
}
    \label{fig:recentPast}
\end{figure}

\begin{figure}[htb]
    \centering
\scalebox{0.45}{\begin{tikzpicture}[block/.style={draw, rectangle, minimum size=0.2cm, align=center, fill=gray!30},
                    tipset/.style={draw, ellipse, inner sep=2pt, fit=#1, thick},
                    line/.style={draw, thick, ->},
                    orphan/.style={draw, rectangle, minimum size=0.2cm, align=center, fill=red!50},
                    futblock/.style={draw, rectangle, minimum size=0.2cm, align=center, fill=blue!50},
                    tipsetalt/.style={draw, ellipse, inner sep=2pt, fit=#1, thick, dashed, red},
                    tipsetfut/.style={draw, ellipse, inner sep=2pt, fit=#1, thick, dashed, blue},
                    tipsethighlight/.style={draw, ellipse, inner sep=2pt, fit=#1, very thick},
                    chainline/.style={draw, thick, dashed, ->, red},
                    futline/.style={draw, thick, dashed, ->, blue},
                    compcloud/.style={draw, cloud, cloud puffs=11, cloud ignores aspect, inner sep=-4pt, fit=#1, very thick, red},
                    futcloud/.style={draw, cloud, cloud puffs=11, cloud ignores aspect, inner sep=-4pt, fit=#1, very thick, blue}]

    
    \node[block] at (0,0) (s1) {};
    \node[block, below=0.1cm of s1] (s2) {};
    \node[block, below=0.1cm of s2] (s3) {};
    \node[block, below=0.1cm of s3] (s4) {};
    \node[tipsethighlight={(s1) (s2) (s3) (s4)}, label=above: $s$] (ts) {};
    
    \node[block, right=0.6cm of s1] (b11) {};
    \node[block, below=0.1cm of b11] (b12) {};
    \node[block, below=0.1cm of b12] (b13) {};
    \node[block, below=0.1cm of b13] (b14) {};
    \node[block, below=0.1cm of b14] (b15) {};
    \node[tipset={(b11) (b12) (b13) (b14) (b15)}] (b1) {};

    \node[block, right=0.6cm of b11] (b21) {};
    \node[block, below=0.1cm of b21] (b22) {};
    \node[block, below=0.1cm of b22] (b23) {};
    \node[block, below=0.1cm of b23] (b24) {};
    \node[tipset={(b21) (b22) (b23) (b24)}] (b2) {};

    \node[block, right=0.6cm of b21] (b31) {};
    \node[block, below=0.1cm of b31] (b32) {};
    \node[block, below=0.1cm of b32] (b33) {};
    \node[block, below=0.1cm of b33] (b34) {};
    \node[block, below=0.1cm of b34] (b35) {};
    \node[block, below=0.1cm of b35] (b36) {};
    \node[tipset={(b31) (b32) (b33) (b34) (b35) (b36)}] (b3) {};

    \node[block, right=0.6cm of b31] (b41) {};
    \node[block, below=0.1cm of b41] (b42) {};
    \node[block, below=0.1cm of b42] (b43) {};
    \node[block, below=0.1cm of b43] (b44) {};
    \node[block, below=0.1cm of b44] (b45) {};
    \node[tipset={(b41) (b42) (b43) (b44) (b45)}] (b4) {};

    \node[block, right=0.6cm of b41] (b51) {};
    \node[block, below=0.1cm of b51] (b52) {};
    \node[block, below=0.1cm of b52] (b53) {};
    \node[block, below=0.1cm of b53] (b54) {};
    \node[tipset={(b51) (b52) (b53) (b54)}] (b5) {};
    
    \node[block, right=0.6cm of b51] (c1) {};
    \node[block, below=0.1cm of c1] (c2) {};
    \node[block, below=0.1cm of c2] (c3) {};
    \node[block, below=0.1cm of c3] (c4) {};
    \node[block, below=0.1cm of c4] (c5) {};
    \node[tipsethighlight={(c1) (c2) (c3) (c4) (c5)}, label=above: $c$] (tc) {};

    \node[futblock, right=0.6cm of c1] (bp11) {};
    \node[futblock, below=0.1cm of bp11] (bp12) {};
    \node[futblock, below=0.1cm of bp12] (bp13) {};
    \node[tipsetfut={(bp11) (bp12) (bp13)}] (bp1) {};

    \node[futblock, right=0.6cm of bp11] (bp21) {};
    \node[futblock, below=0.1cm of bp21] (bp22) {};
    \node[futblock, below=0.1cm of bp22] (bp23) {};
    \node[futblock, below=0.1cm of bp23] (bp24) {};
    \node[tipsetfut={(bp21) (bp22) (bp23) (bp24)}] (bp2) {};

    \node[futblock, right=0.6cm of bp21] (bp31) {};
    \node[futblock, below=0.1cm of bp31] (bp32) {};
    \node[futblock, below=0.1cm of bp32] (bp33) {};
    \node[tipsetfut={(bp31) (bp32) (bp33)}] (bp3) {};

    \node[futblock, right=0.6cm of bp31] (bp41) {};
    \node[futblock, below=0.1cm of bp41] (bp42) {};
    \node[tipsetfut={(bp41) (bp42)}] (bp4) {};

    \node[futblock, right=0.6cm of bp41] (bp51) {};
    \node[futblock, below=0.1cm of bp51] (bp52) {};
    \node[futblock, below=0.1cm of bp52] (bp53) {};
    \node[tipsetfut={(bp51) (bp52) (bp53)}] (bp5) {};

    \node [left=0.6cm of ts] (dots) {\ldots};

    \node [right=0.6cm of bp5] (fut_dots) {\ldots};

    \node[orphan, above=1.9cm of bp11] (f11) {};
    \node[orphan, above=0.1cm of f11] (f12) {};
    \node[tipsetalt={(f11) (f12)}] (f1) {};
    
    \node [left=0.6cm of f1, text=red] (fut_past_dots) {\ldots};

    \node[orphan, right=0.6cm of f11] (f21) {};
    \node[orphan, above=0.1cm of f21] (f22) {};
    \node[orphan, above=0.1cm of f22] (f23) {};
    \node[tipsetalt={(f21) (f22) (f23)}] (f2) {};

    \node[orphan, right=0.6cm of f21] (f31) {};
    \node[orphan, above=0.1cm of f31] (f32) {};
    \node[orphan, above=0.1cm of f32] (f33) {};
    \node[orphan, above=0.1cm of f33] (f34) {};
    \node[orphan, above=0.1cm of f34] (f35) {};
    \node[tipsetalt={(f31) (f32) (f33) (f34) (f35)}] (f3) {};

    \node[orphan, right=0.6cm of f31] (f41) {};
    \node[orphan, above=0.1cm of f41] (f42) {};
    \node[orphan, above=0.1cm of f42] (f43) {};
    \node[tipsetalt={(f41) (f42) (f43)}] (f4) {};
    
    \node[orphan, right=0.6cm of f41] (f51) {};
    \node[orphan, above=0.1cm of f51] (f52) {};
    \node[orphan, above=0.1cm of f52] (f53) {};
    \node[orphan, above=0.1cm of f53] (f54) {};
    \node[tipsetalt={(f51) (f52) (f53) (f54)}] (f5) {};
    
    \draw[line] (ts) -- (dots);
    \draw[line] (b1) -- (ts);
    \draw[line] (b2) -- (b1);
    \draw[line] (b3) -- (b2);
    \draw[line] (b4) -- (b3);
    \draw[line] (b5) -- (b4);
    \draw[line] (tc) -- (b5);
    \draw[futline] (bp1) -- (tc);
    \draw[futline] (bp2) -- (bp1);
    \draw[futline] (bp3) -- (bp2);
    \draw[futline] (bp4) -- (bp3);
    \draw[futline] (bp5) -- (bp4);
    \draw[futline] (fut_dots) -- (bp5);
    
    \draw[chainline] (f1) -- (fut_past_dots);
    \draw[chainline] (f2) -- (f1);
    \draw[chainline] (f3) -- (f2);
    \draw[chainline] (f4) -- (f3);
    \draw[chainline] (f5) -- (f4);

    \node[compcloud={(f1) (f2) (f3) (f4) (f5)}, label={[text=red]above:$M=17-15=2$}] (compcloud) {};

    \node[futcloud={(bp1) (bp2) (bp3) (bp4) (bp5)}, label={[label distance=0.35cm, text=blue]below:slowed growth}] (futcloud) {};
    
    \draw [decorate,decoration={brace,amplitude=10pt,mirror,raise=5pt}, draw=green!40!black]
    ($(ts.south west) - (0,1)$) -- ($(tc.south east) - (0,0.7)$) 
    node [pos=0.5,anchor=north,yshift=-0.55cm, text=green!40!black] {good addition$=33$};
\end{tikzpicture}}
\caption{A diagram illustrating the meaning of~$M$ when accounting for the future.
}
    \label{fig:future}
\end{figure}

\end{document}